\documentclass[11pt,a4paper]{article}

\usepackage[margin=1in]{geometry}
\usepackage{amsmath,amssymb,thmtools, amsthm, booktabs,color,doi,graphicx,latexsym,url,xcolor,xspace} 
\usepackage{comment}

\usepackage[numbers,sort&compress]{natbib}
\usepackage{microtype,hyperref}
\usepackage{eucal}
\usepackage[inline]{enumitem}
\setlist[itemize]{label=--}
\setlist[enumerate]{label=(\arabic*),labelindent=\parindent,leftmargin=*}
\usepackage[capitalize]{cleveref}

\newcommand{\denseexp}{1.158}
\newcommand{\denseexpsr}{4/3}
\newcommand{\mainresult}{1.907}
\newcommand{\mainresultsr}{1.927}
\newcommand{\interfacedensity}{1.814}
\newcommand{\interfacedensitysr}{1.854}
\newcommand{\epschoice}{0.186}
\newcommand{\epschoicesr}{0.146}
\newcommand{\eps}{\varepsilon}
\newcommand{\E}{\mathbb{E}}
\newcommand{\indA}{\hat{A}}
\newcommand{\indB}{\hat{B}}
\newcommand{\indX}{\hat{X}}
\newcommand{\TT}{\mathcal{T}}
\newcommand{\TTA}{\TT_{\operatorname{I}}}
\newcommand{\TTB}{\TT_{\operatorname{II}}}
\newcommand{\PP}{\mathcal{P}}
\newcommand{\TTall}{\hat{\TT}}
\newcommand{\virt}[1]{{#1}'}
\newcommand{\figcaptionstyle}{\boldmath}

\newtheorem{theorem}{Theorem}
\newtheorem{lemma}[theorem]{Lemma}

\newtheorem{definition}[theorem]{Definition}
\newtheorem{observation}[theorem]{Observation}

\definecolor{citecolor}{HTML}{0000C0}
\definecolor{urlcolor}{HTML}{000080}

\hypersetup{
    colorlinks=true,
    linkcolor=black,
    citecolor=citecolor,
    filecolor=black,
    urlcolor=urlcolor,
}

\newenvironment{myabstract}
{\list{}{\listparindent 1.5em%
		\itemindent    \listparindent
		\leftmargin    1cm
		\rightmargin   1cm
		\parsep        0pt}%
	\item\relax}
{\endlist}

\newenvironment{mycover}
{\list{}{\listparindent 0pt
		\itemindent    \listparindent
		\leftmargin    1cm
		\rightmargin   1cm
		\parsep        0pt}%
	\raggedright
	\item\relax}
{\endlist}

\newcommand{\myemail}[1]{\,$\cdot$\, {\small #1}}
\newcommand{\myaff}[1]{\,$\cdot$\, {\small #1}\par\medskip}

\hypersetup{
pdftitle={Sparse Matrix Multiplication in the Low-Bandwidth Model},
pdfauthor={},
}

\begin{document}

\begin{mycover}
	{\huge\bfseries\boldmath Sparse Matrix Multiplication in the Low-Bandwidth Model \par}
	\bigskip
	\bigskip
    \textbf{Chetan Gupta}
    \myemail{chetan.gupta@aalto.fi}
    \myaff{Aalto University}
	
    \textbf{Juho Hirvonen}
    \myemail{juho.hirvonen@aalto.fi}
    \myaff{Aalto University}
	    
    \textbf{Janne H.\ Korhonen}
    \myemail{janne.h.korhonen@gmail.com}
    \myaff{TU Berlin}

    \textbf{Jan Studený}
    \myemail{jan.studeny@aalto.fi}
    \myaff{Aalto University}

    \textbf{Jukka Suomela}
    \myemail{jukka.suomela@aalto.fi}
    \myaff{Aalto University}

\end{mycover}

\medskip
\begin{myabstract}
  \noindent\textbf{Abstract.}
We study matrix multiplication in the low-bandwidth model: There are $n$ computers, and we need to compute the product of two $n \times n$ matrices. Initially computer $i$ knows row $i$ of each input matrix. In one communication round each computer can send and receive one $O(\log n)$-bit message. Eventually computer $i$ has to output row $i$ of the product matrix.

We seek to understand the complexity of this problem in the \emph{uniformly sparse} case: each row and column of each input matrix has at most $d$ non-zeros and in the product matrix we only need to know the values of at most $d$ elements in each row or column. This is exactly the setting that we have, e.g., when we apply matrix multiplication for triangle detection in graphs of maximum degree~$d$. We focus on the \emph{supported} setting: the structure of the matrices is known in advance; only the numerical values of nonzero elements are unknown.

There is a trivial algorithm that solves the problem in $O(d^2)$ rounds, but for a large $d$, better algorithms are known to exist; in the moderately dense regime the problem can be solved in $O(dn^{1/3})$ communication rounds, and for very large $d$, the dominant solution is the fast matrix multiplication algorithm using $O(n^{\denseexp})$ communication rounds (for matrix multiplication over fields and rings supporting fast matrix multiplication).

In this work we show that it is possible to overcome quadratic barrier for \emph{all} values of $d$: we present an algorithm that solves the problem in $O(d^{\mainresult})$ rounds for fields and rings supporting fast matrix multiplication and $O(d^{\mainresultsr})$ rounds for semirings, independent of $n$.
\bigskip
\end{myabstract}

\clearpage

\section{Introduction}

We study the task of multiplying very large but sparse matrices in distributed and parallel settings: there are $n$ computers, each computer knows one row of each input matrix, and each computer needs to output one row of the product matrix. There are numerous efficient matrix multiplication algorithms for dense and moderately sparse matrices, e.g.\ \cite{alg-method-congest-cliq2019,le2016further,DBLP:conf/opodis/Censor-HillelLT18,censor2021fast}---however, these works focus on a high-bandwidth setting, where the problem becomes trivial for very sparse matrices. We instead take a more fine-grained approach and consider \emph{uniformly sparse matrices} in a \emph{low-bandwidth} setting. In this regime, the state-of-the-art algorithm for very sparse matrices is a trivial algorithm that takes  $O(d^2)$ communication rounds; here $d$ is a density parameter we will shortly introduce. In this work we present the \emph{first algorithm that breaks the quadratic barrier}, achieving a running time of $O(d^{\mainresult})$ rounds. We will now introduce the sparse matrix multiplication problem in detail in \cref{ssec:intro-setting}, and then describe the model of computing we will study in this work in \cref{ssec:intro-model}.

\subsection{Setting: uniformly sparse square matrices}\label{ssec:intro-setting}

In general, the complexity of matrix multiplication depends on the sizes and shapes of the matrices, as well as the density of each matrix and the density of the product matrix; moreover, density can refer to e.g.\ the average or maximum number of non-zero elements per row and column. In order to focus on the key challenge---overcoming the quadratic barrier---we will introduce here a very simple version of the sparse matrix multiplication problem, with only two parameters: $n$ and $d$.

We are given three $n \times n$ matrices, $A$, $B$, and $\indX$. Each matrix is uniformly sparse in the following sense: there are at most $d$ non-zero elements in each row, and at most $d$ non-zero elements in each column. Here $A$ and $B$ are our input matrices, and $\indX$ is an indicator matrix. Our task is to compute the matrix product \[X = AB,\] but we only need to report those elements of $X$ that correspond to non-zero elements of $\indX$. That is, we need to output
\[
    X_{ik} = \sum_j A_{ij}B_{jk} \text{ for each $i,k$ with $\indX_{ik} \ne 0$.}
\]
Note that the product $X$ itself does not need to be sparse (indeed, there we might have up to $\Theta(d^2)$ non-zero elements per row and column); it is enough that the set of elements that we care about is sparse.

We will study here both matrix multiplication over rings and matrix multiplication over semirings.

\paragraph{Application: triangle detection and counting.}

While the problem setting is rather restrictive, it captures precisely e.g.\ the widely-studied task of triangle detection and counting in bounded-degree graphs (see e.g.\ \cite{dolev2012tri,izumi2017triangle,10.1145/3460900,10.1145/3446330,chang2019improved,korhonen2017deterministic,10.1145/3382734.3405742} for prior work related to triangle detection in distributed settings). Assume $G$ is a graph with $n$ nodes, of maximum degree $d$, represented as an adjacency matrix; note that this matrix is uniformly sparse. We can set $A = B = \indX = G$, and compute $X = AB$ in the uniformly sparse setting. Now consider an edge $(i,k)$ in $G$; by definition we have $\indX_{ik} \ne 0$, and hence we have computed the value of $X_{ik}$. We will have $X_{ik} \ne 0$ if and only if there exists a triangle of the form $\{i,j,k\}$ in graph $G$; moreover, $X_{ik}$ will indicate the number of such triangles, and given $\indX$ and $X$, we can easily also calculate the total number of triangles in the graph (keeping in mind that each triangle gets counted exactly $6$ times).

\subsection{Supported low-bandwidth model}\label{ssec:intro-model}

\paragraph{Low-bandwidth model.}

We seek to solve the uniformly sparse matrix multiplication problem using $n$ parallel computers, in a message-passing setting. Each computer has got its own local memory, and initially computer number $i$ knows row $i$ of each input matrix. Computation proceeds in synchronous communication rounds, and in each round each computer can send one $O(\log n)$-bit message to another computer and receive one such message from another computer (we will assume that the elements of the ring or semiring over which we do matrix multiplication can be encoded in such messages). Eventually, computer number $i$ will have to know row $i$ of the product matrix (or more precisely, those elements of the row that we care about, as indicated by row $i$ of matrix $\indX$). We say that an algorithm runs in time $T$ if it stops after $T$ communication rounds (note that we do not restrict local computation here; the focus is on communication, which tends to be the main bottleneck in large-scale computer systems).

Recent literature has often referred to this model as the \emph{node-capacitated clique} \cite{node-capacitated2019} or \emph{node-congested clique}. It is also a special case of the classical bulk synchronous parallel model \cite{BSP1990}, with local computation considered free. In this work we will simply call this model the \emph{low-bandwidth model}, to highlight the key difference with e.g.\ the \emph{congested clique} model \cite{congest-clique2003} which, in essence, is the same model with $n$ times more bandwidth per node per round.

\paragraph{Supported version.}

To focus on the key issue here---the quadratic barrier---we will study the \emph{supported} version of the low-bandwidth model: we assume that the structure of the input matrices is known in advance. More precisely, we know in advance uniformly sparse indicator matrices $\indA$, $\indB$, and $\indX$, where $\indA_{ij} = 0$ implies that $A_{ij} = 0$, $\indB_{jk} = 0$ implies that $B_{jk} = 0$, and $\indX_{ik} = 0$ implies that we do not need to calculate $X_{ik}$. However, the values of $A_{ij}$ and $B_{jk}$ are only revealed at run time.

In the triangle counting application, the supported model corresponds to the following setting: We have got a fixed, globally known graph $G$, with degree at most $d$. Edges of $G$ are colored either red or blue, and we need to count the number of triangles formed by red edges; however, the information on the colors of the edges is only available at the run time. Here we can set $\indA = \indB = \indX = G$, and then the input matrix $A = B$ will indicate which edges are red.

We note that, while the supported version is \emph{a priori} a significant strengthening of the low-bandwidth model, it is known that e.g.~the supported version of the CONGEST model is not significantly stronger than baseline CONGEST, and almost all communication complexity lower bounds for CONGEST are also known to apply for supported CONGEST~\cite{foerster2019preprocessing}.

\subsection{Contributions and prior work}\label{ssec:intro-contrib}

The high-level question we set to investigate here is as follows: what is the round complexity of uniformly sparse matrix multiplication in the supported low-bandwidth model, as a function of parameters $n$ and $d$. \Cref{fig:complexity} gives an overview of what was known by prior work and what is our contribution; the complexity of the problem lies in the shaded region.

\begin{figure}
\includegraphics[page=1,scale=0.8]{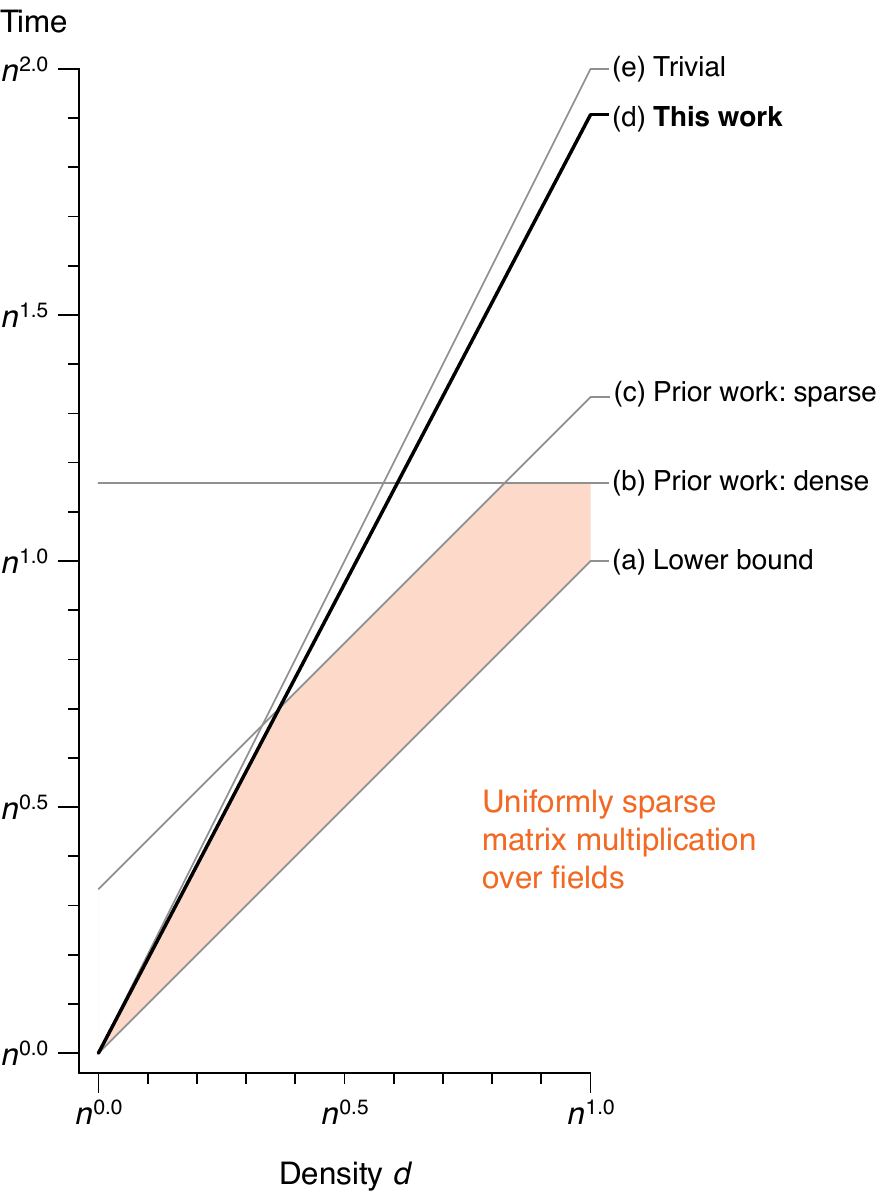}%
\hspace*{\stretch{1}}%
\includegraphics[page=2,scale=0.8]{figs.pdf}

\caption{Complexity of sparse matrix multiplication in the low-bandwidth model: prior work and the new result.}\label{fig:complexity}
\end{figure}

The complexity has to be at least $\Omega(d)$ rounds, by a simple information-theoretic argument (in essence, all $d$ units of information held by a node have to be transmitted to someone else); this gives the lower bound marked with (a) in \cref{fig:complexity}. Above $\Omega(d)$ the problem definition is also robust to minor variations (e.g., it does not matter whether element $A_{ij}$ is initially held by node $i$, node $j$, or both, as in $O(d)$ additional rounds we can transpose a sparse matrix).

The complexity of dense matrix multiplication over rings in the low-band\-width model is closely connected to the complexity of matrix multiplication with centralized sequential algorithms: if there is a centralized sequential algorithm that solve matrix multiplication with $O(n^{\omega})$ element-wise multiplications, there is also an algorithm for the congested clique model that runs in $O(n^{1-2/\omega})$ rounds \cite{alg-method-congest-cliq2019}, and this can be simulated in the low-bandwidth model in $O(n^{2-2/\omega})$ rounds. For fields and at least certain rings such as integers, we can plug in the latest bound for matrix multiplication exponent $\omega < 2.3728596$ \cite{alman2021refined} to arrive at the round complexity of $O(n^{\denseexp})$, illustrated with line~(b) in the figure. For semirings the analogous result is $O(n^{\denseexpsr})$ rounds.

For sparse matrices we can do better by using the algorithm from \cite{censor2021fast}. This algorithm is applicable for both rings and semirings, and in our model the complexity is $O(dn^{1/3})$ rounds; this is illustrated with line~(c). However, for small values of $d$ we can do much better even with a trivial algorithm where node $j$ sends each $B_{jk}$ to every node $i$ with $A_{ij} \ne 0$; then node $i$ can compute $X_{ik}$ for all $k$. Here each node sends and receives $O(d^2)$ values, and this takes $O(d^2)$ rounds (see \cref{lem:brute-force} for the details). We arrive at the upper bound shown in line~(e).

In summary, the problem can be solved in $O(d^2)$ rounds, independent of $n$. However, when $d$ increases, we have got better upper bounds, and for $d = n$ we eventually arrive at $O(n^{\denseexp})$ or $O(n^{\denseexpsr})$ instead of the trivial bound $O(n^2)$.

Now the key question is if we can break the quadratic barrier $O(d^2)$ for all values of $d$. For example, could one achieve a bound like $O(d^{1.5})$? Or, is there any hope of achieving a bound like $O(d^{\denseexp})$ or $O(d^{\denseexpsr})$ rounds for all values of $d$?

We prove that the quadratic barrier can be indeed broken; our new upper bound is shown with line~(d) in \cref{fig:complexity}:
\begin{theorem}\label{thm:main-intro}
There is an algorithm that solves uniformly sparse matrix multiplication over fields and rings supporting fast matrix multiplication in $O(d^{\mainresult})$ rounds in the supported low-band\-width model.
\end{theorem}
It should be noted that the value of the matrix multiplication exponent $\omega$ can depend on the ground field or ring we are operating over. The current best bound $\omega < 2.3728596$ \cite{alman2021refined} holds over any field and at least certain rings such as integers, and \cref{thm:main-intro} should be understood as tacitly referring to rings for which this holds. More generally, for example Strassen's algorithm~\cite{strassen1969gaussian} giving $\omega < 2.8074$ can be used over any ring, yielding a running time of $O(d^{1.923})$ rounds using techniques of \cref{thm:main-intro}. We refer interested readers to \cite{censor2021fast} for details of translating matrix multiplication exponent bounds to congested clique algorithms, and to \cite{burgisser-algebraic} for more technical discussion on matrix multiplication exponent in general.

We can also break the quadratic barrier for arbitrary semirings, albeit with a slightly worse exponent:
\begin{theorem}\label{thm:main-intro-sr}
There is an algorithm that solves uniformly sparse matrix multiplication over semirings in $O(d^{\mainresultsr})$ rounds in the supported low-band\-width model.
\end{theorem}

We see our work primarily as a proof of concept for breaking the barrier; there is nothing particularly magic about the specific exponent $\mainresult$, other than that it demonstrates that values substantially smaller than $2$ can be achieved---there is certainly room for improvement in future work, and one can verify that even with $\omega = 2$ we do not get $O(d)$ round complexity. Also, we expect that the techniques that we present here are applicable also beyond the specific case of supported low-bandwidth model and uniformly sparse matrices.

\subsection{Open questions}

For future work, there are four main open question:
\begin{enumerate}
    \item What is the smallest $\alpha$ such that uniformly sparse matrix multiplication can be solved in $O(d^{\alpha})$ rounds in the supported low-bandwidth model?
    \item Can we eliminate the assumption that there is a support (i.e., the structure of the matrices is known)?
	\item Can we use fast uniformly sparse matrix multiplication to obtain improvements for the general sparse case, e.g.~by emulating the Yuster--Zwick algorithm~\cite{10.1145/1077464.1077466}?
    \item Can the techniques that we introduce here be applied also in the context of the CONGEST model (cf.\ \cite{10.1145/3446330,chang2019improved,izumi2017triangle,korhonen2017deterministic})?
\end{enumerate}

\section{Proof overview and key ideas}\label{sec:proof-overview}

Even though the task at hand is about linear algebra, it turns out that it is helpful to structure the algorithm around graph-theoretic concepts.

\subsection{Nodes, triangles, and graphs}

In what follows, it will be convenient to view our input and our set of computers as having a \emph{tripartite} structure. Let $I$, $J$, and $K$ be disjoint sets of size $n$; we will use these sets to index the matrices $A$, $B$, and $X$ so that the elements are $A_{ij}$, $B_{jk}$, and $X_{ik}$ for $i \in I$, $j \in J$, and $k \in K$. Likewise, we use the sets $I$, $J$, and $K$ to index the indicator matrices $\indA$, $\indB$, and $\indX$ given to the computers in advance. We will collectively refer to $V = I \cup J \cup K$ as the set of \emph{nodes}. We emphasize that we have got $|V| = 3n$, so let us be careful not to confuse $|V|$ and $n$.

Concretely, we interpret this so that we are computing an $n \times n$ matrix product using $3n$ computers---each node $v \in V$ is a computer, such that node $i \in I$ initially knows $A_{ij}$ for all $j$, node $j \in J$ initially knows $B_{jk}$ for all $k$, and node $i \in I$ needs to compute $X_{ik}$ for all $k$. In our model of computing we had only $n$ computers, but we can transform our problem instance into the tripartite formalism by having one physical computer to simulate $3$ virtual computers in a straightforward manner. This simulation only incurs constant-factor overhead in running times, so we will henceforth assume this setting as given.

Now we are ready to introduce the key concept we will use throughout this work:

\begin{definition}
Let $i \in I$, $j \in J$, and $k \in K$. We say that $\{i,j,k\}$ is a \emph{triangle} if $\indA_{ij} \ne 0$, $\indB_{jk} \ne 0$, and $\indX_{ik} \ne 0$. We write $\TTall$ for the set of all triangles.
\end{definition}

In other words, a triangle $\{ i, j, k \}$ corresponds to a possibly non-zero product $A_{ij}B_{jk}$ included in an output $X_{ik}$ we need to compute.

Let $\TT \subseteq \TTall$ be a set of triangles. We write $G(\TT)$ for the graph $G(\TT) = (V,E)$, where $E$ consists of all edges $\{u,v\}$ such that $\{u,v\} \subseteq T$ for some triangle $T \in \TT$. As the matrices are uniformly sparse, we have the following simple observations:

\begin{observation}\label{obs:GT-max-deg}
Each node $i \in I$ in $G(\TTall)$ is adjacent to at most $d$ nodes of $J$ and at most $d$ nodes of $K$. A similar observation holds for the nodes of $J$ and $K$. In particular, the maximum degree of $G(\TTall)$ is at most $2d$, and hence the maximum degree of $G(\TT)$ for any $\TT$ is also at most $2d$.
\end{observation}

\begin{observation}\label{obs:GT-size}
Each node $i \in I$ can belong to at most $d^2$ triangles, and hence the total number of triangles in $\TTall$ is at most $d^2n$.
\end{observation}

\begin{observation}\label{obs:GT-JK-edges}
There are at most $dn$ edges between $J$ and $K$ in graph $G(\TTall)$.
\end{observation}

Note that $\TTall$ only depends on $\indA$, $\indB$, and $\indX$, which are known in the supported model, and it is independent of the values of $A$ and~$B$.

\subsection{Processing triangles}

We initialize $X_{ik} \gets 0$; this is a variable held by the computer responsible for node $i \in I$. We say that we have \emph{processed} a set of triangles $\TT \subseteq \TTall$ if the current value of $X_{ik}$ equals the sum of products $a_{ij} b_{jk}$ over all $j$ such that $\{i,j,k\} \in \TT$.

By definition, we have solved the problem if we have processed all triangles in $\TTall$. Hence all that we need to do is to show that all triangles can be processed in $O(d^{\mainresult})$ rounds.

\subsection{Clustering triangles}

The following definitions are central in our work; see \cref{fig:clustering} for an illustration:

\begin{definition}\label{def:cluster}
A set of nodes $U \subseteq V$ is a \emph{cluster} if it consists of $d$ nodes from $I$, $d$ nodes from $J$, and $d$ nodes from $K$.
\end{definition}

If $\TT \subseteq \TTall$ is a set of triangles and $U$ is a cluster, we write
\[
    \TT[U] = \{ T \in \TT : T \subseteq U \}
\]
for the set of triangles contained in $U$.

\begin{definition}\label{def:clustered}
A collection of triangles $\PP \subseteq \TTall$ is \emph{clustered} if there are disjoint clusters $U_1, \dotsc, U_k \subseteq V$ such that $\PP = \PP[U_1] \cup \dotsb \cup \PP[U_k]$.
\end{definition}

\begin{figure}
\centering
\includegraphics[page=3,scale=0.9]{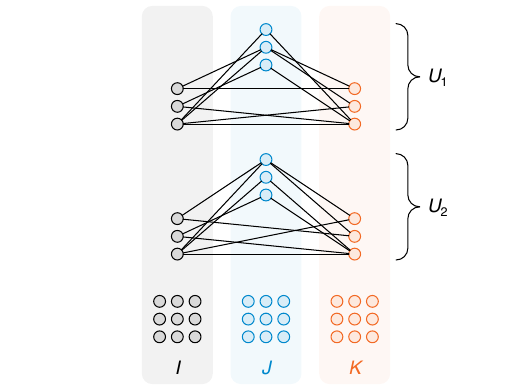}
\caption{\figcaptionstyle A collection of triangles that is clustered, for $d = 3$; in this example we have two clusters, $U_1$ and $U_2$. Note that there can be nodes that are not part of any cluster, but each triangle has to be contained in exactly one cluster.}\label{fig:clustering}
\end{figure}

That is, $\PP$ is clustered if the triangles of $\PP$ can be partitioned into small node-disjoint tripartite structures. The key observation is this:
\begin{lemma}\label{lem:intro-clustered}
    For matrix multiplication over rings, if $\PP \subseteq \TTall$ is clustered, then all triangles in $\PP$ can be processed in $O(d^{\denseexp})$ rounds.
\end{lemma}
\begin{proof}
    Let $U_1, \dotsc, U_k \subseteq V$ denote the clusters (as in \cref{def:clustered}). The task of processing $\PP[U_i]$ using the computers of $U_i$ is equivalent to a dense matrix multiplication problem in a network with $3d$ nodes. Hence each subset of nodes $U_i$ can run the dense matrix multiplication algorithm from \cite{alg-method-congest-cliq2019} in parallel, processing all triangles of $\PP$ in $O(d^{\denseexp})$ rounds.
\end{proof}

By applying the dense matrix multiplication algorithm for semiring from \cite{alg-method-congest-cliq2019}, we also get:
\begin{lemma}\label{lem:intro-clustered-sr}
    For matrix multiplication over semirings, if $\PP \subseteq \TTall$ is clustered, then all triangles in $\PP$ can be processed in $O(d^{\denseexpsr})$ rounds.
\end{lemma}

\subsection{High-level idea}\label{ssec:high-level}

Now we are ready to present the high-level idea; see \cref{fig:overview} for an illustration. We show that any $\TTall$ can be partitioned in two components, $\TTall = \TTA \cup \TTB$, where $\TTA$ has got a nice structure that makes it easy to process, while $\TTB$ is small. The more time we put in processing $\TTA$, the smaller we can make $\TTB$, and small sets are fast to process:
\begin{itemize}
    \item For matrix multiplication over rings, if we spend $O(d^{\mainresult})$ rounds in the first phase to process $\TTA$, we can ensure that $\TTB$ contains only $O(d^{\interfacedensity}n)$ triangles, and then it can be also processed in $O(d^{\mainresult})$ rounds. \Cref{thm:main-intro} follows.
    \item For matrix multiplication over semirings, if we spend $O(d^{\mainresultsr})$ rounds in the first phase to process $\TTA$, we can ensure that $\TTB$ contains only $O(d^{\interfacedensitysr}n)$ triangles, which can be processed in $O(d^{\mainresultsr})$ rounds. \Cref{thm:main-intro-sr} follows.
\end{itemize}

\begin{figure}
\centering
\includegraphics[page=4,scale=0.9]{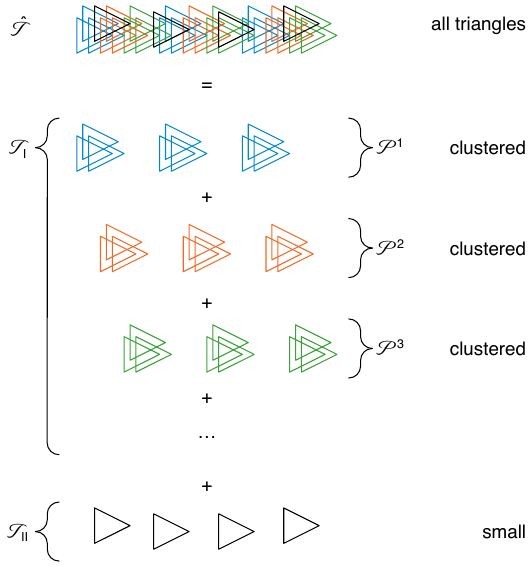}
\caption{\figcaptionstyle Any set of triangles $\TTall$ can be decomposed in two components, a clustered component $\TTA$ and a small component $\TTB$. We can process $\TTA$ efficiently by applying dense matrix multiplication inside each cluster, and we can process $\TTB$ efficiently since it is small.}\label{fig:overview}
\end{figure}

We will now explain how to construct and process $\TTA$ and $\TTB$ in a bit more detail. We emphasize that the \emph{constructions} of $\TTA$ and $\TTB$ only depends on $\TTall$, and hence all of this preparatory work can be precomputed in the supported model. Only the actual \emph{processing} of $\TTA$ and $\TTB$ takes place at run time.
We will use the case of rings (\cref{thm:main-intro}) here as an example; the case of semirings (\cref{thm:main-intro-sr}) follows the same basic idea.

\paragraph{Component $\TTA$ is large but clustered.}

To construct $\TTA$, we start with $\TT^1 = \TTall$. Then we repeatedly choose a clustered subset $\PP^i \subseteq \TT^i$, and let $\TT^{i+1} = \TT^i \setminus \PP^i$. Eventually, after some $L$ steps, $\TT^{L+1}$ will be sufficiently small, and we can stop.

As each set $\PP^i$ is clustered, by \cref{lem:intro-clustered} we can process each of them in $O(d^{\denseexp})$ rounds, and hence the overall running time will be $O(L d^{\denseexp})$ rounds. We choose a small enough $L$ so that the total running time is bounded by $O(d^{\mainresult})$, as needed.

This way we have processed $\TTA = \PP^1 \cup \dotsb \cup \PP^L$. We will leave the remainder $\TTB = \TT^{L+1}$ for the second phase.

For all of this to make sense, the sets $\PP^i$ must be sufficiently large so that we make rapid progress towards a small remainder $\TTB$. Therefore the key nontrivial part of our proof is a graph-theoretic result that shows that if $\TT$ is sufficiently large, we can always find a large clustered subset $\PP \subseteq \TT$. To prove this, we first show that if $\TT$ is sufficiently large, we can find a cluster $U_1$ that contains many triangles. Then we discard $U_1$ and all triangles touching it, and repeat the process until $\TT$ becomes small enough; this way we have iteratively discovered a large clustered subset \[\PP = \PP[U_1] \cup \dotsb \cup \PP[U_k] \subseteq \TT,\] together with disjoint clusters $U_1 \cup \dotsb \cup U_k \subseteq V$. For the details of this analysis we refer to \cref{sec:first-component}.

\paragraph{Component $\TTB$ is small.}

Now we are left with only a relatively small collection of triangles $\TTB$; we chose the parameters so that the total number of triangles in $\TTB$ is $O(d^{\interfacedensity}n)$. We would like to now process all of $\TTB$ in $O(d^{\mainresult})$ rounds.

The key challenge here is that the structure of $\TTB$ can be awkward: even though the average number of triangles per node is low, there might be some nodes that are incident to a large number of triangles. We call nodes that touch too many triangles \emph{bad nodes}, and triangles that touch bad nodes are \emph{bad triangles}.

We first show how we can process all good triangles. This is easy: we can, in essence, make use of the trivial brute-force algorithm.

Then we will focus on the bad triangles. The key observation is that there are only a few bad nodes. If we first imagined that each bad node tried to process its own share of bad triangles, a vast majority of the nodes in the network would be idle. Hence each bad node is able to recruit a large number of \emph{helper nodes}, and with their assistance we can process also all bad triangles sufficiently fast. We refer to \cref{sec:sparse-case} for the details.

\section{Finding one cluster}
\label{sec:find-one-cluster}

Now we will introduce the key technical tool that we will use to construct a decomposition $\TTall = \TTA \cup \TTB$, where $\TTA$ is clustered and $\TTB$ is small. In this section we will show that given any sufficiently large collection of triangles $\TT \subseteq \TTall$, we can always find \emph{one} cluster $U \subseteq V$ that contains many triangles (recall \cref{def:cluster}). The idea is that we will then later apply this lemma iteratively to construct the clustered sets $\PP^1 \cup \dotsb \cup \PP^L = \TTA$.

\begin{lemma}
\label{lem:one-cluster}
Assume that $n \ge d$ and $\eps \ge 0$. Let $\TT \subseteq \TTall$ be a collection of triangles with
\[
     \bigl|\TT\bigr| \ge d^{2-\eps} n.
\]
Then there exists a cluster $U \subseteq V$ with
\[
    \bigl|\TT[U]\bigr| \ge \frac{1}{24} d^{3-4\eps}.
\]
\end{lemma}

Before we prove this lemma, it may be helpful to first build some intuition on this claim. When $\eps = 0$, the assumption is that there are $d^2n$ triangles in $\TT$. Recall that this is also the largest possible number of triangles (\cref{obs:GT-size}). One way to construct a collection with that many triangles is to split the $3n$ nodes in $V$ into $n/d$ clusters, with $3d$ nodes in each, and then inside each cluster we can have $d^3$ triangles. But in this construction we can trivially find a cluster $U \subseteq V$ with $|\TT[U]| = d^3$; indeed, the entire collection of triangles is clustered.

Now one could ask if there is some clever way to construct a large collection of triangles that cannot be clustered. What \cref{lem:one-cluster} shows is that the answer is no: as soon as you somehow put $d^2n$ triangles into collection $\TT$ (while respecting the assumption that the triangles are defined by some uniformly sparse matrices, and hence $G(\TT)$ is of degree at most $d$) you cannot avoid creating at least one cluster that contains $\Omega(d^3)$ triangles. And a similar claim then holds also for slightly lower numbers of triangles, e.g.\ if the total number of triangles is $d^{1.99}n$, we show that you can still find a cluster with $\Omega(d^{2.96})$ triangles.

Let us now prove the lemma. Our proof is constructive; it will also give a procedure for finding such a cluster.

\begin{proof}[Proof of \cref{lem:one-cluster}]
Consider the tripartite graph $G(\TT)$ defined by collection $\TT$. Let $\{j,k\}$ be an edge with $j \in J$ and $k \in K$. We call edge $\{j,k\}$ \emph{heavy} if there are at least $\frac{1}{2}d^{1-\eps}$ triangles $T$ with $\{j,k\} \subseteq T$. We call a triangle \emph{heavy} if its $J$--$K$ edge is heavy.

By \cref{obs:GT-JK-edges} there can be at most $dn$ non-heavy $J$--$K$ edges, and by definition each of them can contribute to at most $\frac{1}{2}d^{1-\eps}$ triangles. Hence the number of non-heavy triangles is at most $\frac{1}{2}d^{2-\eps}n$, and therefore the number of heavy triangles is at least $\frac{1}{2}d^{2-\eps}n$.

Let $\TT_0 \subseteq \TT$ be the set of heavy triangles. For the remainder of the proof, we will study the properties of this subset and the tripartite graph $G(\TT_0)$ defined by it. So from now on, all triangles are heavy, and all edges refer to the edges of $G(\TT_0)$.

First, we pick a node $x \in I$ that touches at least $\frac{1}{2}d^{2-\eps}$ triangles; such a node has to exists as on average each node touches at least $\frac{1}{2}d^{2-\eps}$ triangles.
Let $J_0$ be the set of $J$-corners and $K_0$ be the set of $K$-corners of these triangles. By \cref{obs:GT-max-deg}, we have $|J_0| \le d$ and $|K_0| \le d$.
We label all nodes $i \in I$ by the following values:
\begin{itemize}
    \item $t(i)$ is the number of triangles of the form $i$--$J_0$--$K_0$,
    \item $y(i)$ is the number of $i$--$J_0$ edges,
    \item $z(i)$ is the number of $i$--$K_0$ edges,
    \item $e(i) = y(i) + z(i)$ is the total number of edges from $i$ to $J_0 \cup K_0$.
\end{itemize}

By the choice of $x$, we have got at least $\frac{1}{2}d^{2-\eps}$ triangles of the form $x$--$J_0$--$K_0$. Therefore there are also at least $\frac{1}{2}d^{2-\eps}$ edges of the form $J_0$--$K_0$, and all of these are heavy, so each of them is contained in at least $\frac{1}{2}d^{1-\eps}$ triangles. We have
\begin{equation}
    \sum_{i\in I} t(i) \ge \frac{1}{2}d^{2-\eps} \cdot \frac{1}{2}d^{1-\eps} \ge \frac{1}{4}d^{3-2\eps}. \label{eq:dense-subset-inequalities}
\end{equation}
Since $J_0$ and $K_0$ have at most $d$ nodes, each of degree at most $d$, there are at most $d^2$ edges of the form $I$--$J_0$, and at most $d^2$ edges of the form $I$--$K_0$. Therefore
\begin{equation}
    \sum_{i\in I} y(i) \le d^2, \quad
    \sum_{i\in I} z(i) \le d^2, \quad
    \sum_{i\in I} e(i) \le 2d^2. \label{eq:bound-for-e-values}
\end{equation}
For each triangle of the form $i$--$J_0$--$K_0$ there has to be an edge $i$--$J_0$ and an edge $i$--$K_0$, and hence for each $i \in I$ we have got
\begin{equation}\label{eq:t-ee-bound}
t(i) \le y(i) z(i) \le \frac{e(i)^2}{4}.
\end{equation}

Let $I_0 \subseteq I$ consists of the $d$ nodes with the largest $t(i)$ values (breaking ties arbitrarily), and let $I_1 = I \setminus I_0$. Define
\[
    T_0 = \sum_{i \in I_0} t(i), \quad
    T_1 = \sum_{i \in I_1} t(i), \quad
    t_0 = \min_{i \in I_0} t(i).
\]

First assume that
\begin{equation}\label{eq:T0-assumption}
    T_0 < \frac{1}{24} d^{3-4\eps}.
\end{equation}
Then
\begin{equation}\label{eq:t0-assumption}
    t_0 \le \frac{T_0}{d} < \frac{1}{24} d^{2-4\eps}.
\end{equation}
By definition, $t(i) \le t_0$ for all $i \in I_1$. By \eqref{eq:t-ee-bound} we have got
\begin{equation}\label{eq:I1ti}
    t(i)
    \le \sqrt{t(i)}\cdot \frac{e(i)}{2}
    \le \frac{\sqrt{t_0}}{2} \cdot e(i)
\end{equation}
for all $i \in I_1$. But from \eqref{eq:bound-for-e-values} we have
\begin{equation}\label{eq:I1ei}
    \sum_{i\in I_1} e(i) \le \sum_{i\in I} e(i) \le 2d^2.
\end{equation}
By putting together  \eqref{eq:t0-assumption}, \eqref{eq:I1ti}, and \eqref{eq:I1ei}, we get
\[
    T_1 = \sum_{i \in I_1} t(i) \le \sum_{i\in I_1} \frac{\sqrt{t_0}}{2} e(i) \le \sqrt{t_0} d^2 < \frac{1}{\sqrt{24}} d^{3-2\eps}.
\]
But we also have from \eqref{eq:T0-assumption} that
\[
    T_0 < \frac{1}{24} d^{3-4\eps} \le \frac{1}{24} d^{3-2\eps},
\]
and therefore we get
\[
    \sum_{i\in I} t(i) = T_0 + T_1 < \biggl(\frac{1}{24} + \frac{1}{\sqrt{24}}\biggr) d^{3-2\eps} < \frac{1}{4} d^{3-2\eps},
\]
but this contradicts \eqref{eq:dense-subset-inequalities}.

Therefore we must have
\[
    T_0 \ge \frac{1}{24} d^{3-4\eps}.
\]
Recall that there are exactly $T_0$ triangles of the form $I_0$--$J_0$--$K_0$, set $I_0$ contains by construction exactly $d$ nodes, and $J_0$ and $K_0$ contain at most $d$ nodes. Since we had $n \ge d$, we can now add arbitrarily nodes from $J$ to $J_0$ and nodes from $K$ to $K_0$ so that each of them has size exactly $d$. Then $U = I_0 \cup J_0 \cup K_0$ is a cluster with
\[
    \bigl|\TT[U]\bigr| \ge \bigl|\TT_0[U]\bigr| \ge T_0 \ge \frac{1}{24} d^{3-4\eps}.
    \qedhere
\]
\end{proof}

\section{Finding many disjoint clusters}\label{sec:first-component}

In \cref{sec:find-one-cluster} we established our key technical tool. We are now ready to start to follow the high-level plan explained in \cref{ssec:high-level,fig:overview}. Recall that the plan is to partition $\TTall$ in two parts, $\TTA$ and $\TTB$, where $\TTA$ has got a nice clustered structure and the remaining part $\TTB$ is small. In this section we will show how to construct and process $\TTA$. In \cref{sec:sparse-case} we will then see how to process the remaining part $\TTB$.

\subsection{One clustered set}

With the following lemma we can find one large clustered subset $\PP \subseteq \TT$ for any sufficiently large collection of triangles $\TT$. In \cref{fig:overview}, this corresponds to the construction of, say, $\PP^1$.

\begin{lemma}
\label{lem:clustering-layer}
Let $\eps_2 \ge 0$ and $\delta > 0$, and assume that $d$ is sufficiently large.
Let $\TT \subseteq \TTall$ be a collection of triangles with
\[
     \bigl|\TT\bigr| \ge d^{2-\eps_2} n.
\]
Then $\TT$ can be partitioned into disjoint sets
\[
    \TT = \PP \cup \TT',
\]
where $\PP$ is clustered and
\[
    \bigl|\PP\bigr| \ge \frac{1}{144} d^{2-5\eps_2-4\delta} n.
\]
\end{lemma}

\begin{proof}
We construct $\PP$ and $\TT'$ iteratively as follows. Start with the original collection $\TT$. Then we repeat the following steps until there are fewer than $d^{2-\eps_2-\delta} n$ triangles left in $\TT$:
\begin{enumerate}
    \item Apply \cref{lem:one-cluster} to $\TT$ with $\eps = \eps_2 + \delta$ to find a cluster $U$.
    \item For each triangle $T \in \TT[U]$, add $T$ to $\PP$ and remove it from~$\TT$.
    \item For each triangle $T \in \TT$ with $T \cap U \ne \emptyset$, add $T$ to $\TT'$ and remove it from $\TT$.
\end{enumerate}
Finally, for each triangle $T \in \TT$ that still remains, we add $T$ to $\TT'$ and remove it from $\TT$.

Now by construction, $\PP \cup \TT'$ is a partition of $\TT$. We have also constructed a set $\PP$ that is \emph{almost} clustered: if $U$ shares some node $v$ with a cluster that was constructed earlier, then $v$ is not incident to any triangle of $\TT$, and hence $v$ can be freely replaced with any other node that we have not used so far. Hence we can easily also ensure that $\PP$ is clustered, with minor additional post-processing.

We still need to prove that $\PP$ is large.

Each time we apply \cref{lem:one-cluster}, we delete at most $3d^3$ triangles from $\TT$: there are $3d$ nodes in cluster $U$, and each is contained in at most $d^2$ triangles (\cref{obs:GT-size}). On the other hand, iteration will not stop until we have deleted at least
\[
    d^{2- \eps_2}n - d^{2- \eps_2 - \delta}n \geq \frac{1}{2}d^{2- \eps_2}n
\]
triangles; here we make use of the assumption that $d$ is sufficiently large (in comparison with the constant $\delta$). Therefore we will be able to apply \cref{lem:one-cluster} at least
\[
    \frac{\frac{1}{2}d^{2- \eps_2}n}{3d^3} = \frac{1}{6}d^{-1 - \eps_2}n
\]
times, and each time we add to $\PP$ at least
\[
    \frac{1}{24} d^{3-4\eps_2-4\delta}
\]
triangles, so we have got
\[
    \bigl|\PP\bigr|
    \ge \frac{1}{24} d^{3-4\eps_2-4\delta} \cdot \frac{1}{6} d^{-1 - \eps_2} n
    = \frac{1}{144} d^{2-5\eps_2-4\delta} n.
    \qedhere
\]
\end{proof}

\subsection{Many clusterings}

Next we will apply \cref{lem:clustering-layer} repeatedly to construct all clustered sets $\PP^1, \dotsc, \PP^L$ shown in \cref{fig:overview}.

\begin{lemma}
\label{lem:clustering}
Let $0 \le \eps_1 < \eps_2$ and $\delta > 0$, and assume that $d$ is sufficiently large.
Let $\TT \subseteq \TTall$ be a collection of triangles with
\[
     \bigl|\TT\bigr| \le d^{2-\eps_1} n.
\]
Then $\TT$ can be partitioned into disjoint sets
\[
    \TT = \PP^1 \cup \dotsc \cup \PP^L \cup \TT',
\]
where each $\PP^i$ is clustered, the number of layers is
\[
    L \le 144 d^{5 \eps_2 - \eps_1 + 4\delta},
\]
and the number of triangles in the residual part $\TT'$ is
\[
     \bigl|\TT'\bigr| \le d^{2 - \eps_2} n.
\]
\end{lemma}

\begin{proof}
Let $\TT^1 = \TT$. If
\[
    \bigl|\TT^i \bigr| \le d^{2 - \eps_2} n,
\]
we can stop and set $L = i-1$ and $\TT' = \TT^i$. Otherwise we can apply \cref{lem:clustering-layer} to partition $\TT^i$ into the clustered part $\PP^i$ and the residual part $\TT^{i+1}$.

By \cref{lem:clustering-layer}, the number of triangles in each set $\PP^i$ is at least
\[
    \bigl|\PP^i\bigr| \ge \frac{1}{144} d^{2-5\eps_2-4\delta} n,
\]
while the total number of triangles was by assumption
\[
     \bigl|\TT\bigr| \le d^{2-\eps_1} n.
\]
Hence we will run out of triangles after at most
\[
    L = \frac{d^{2-\eps_1} n}{\frac{1}{144} d^{2-5\eps_2-4\delta} n} = 144 d^{5\eps_2-\eps_1+4\delta}
\]
iterations.
\end{proof}

\subsection{Simplified version}
\label{subsec:simplified}

If we are only interested in breaking the quadratic barrier, we have now got all the ingredients we need to split $\TTall$ into a clustered part $\TTA$ and a small part $\TTB$ such that both $\TTA$ and $\TTB$ can be processed fast.

\begin{lemma}\label{lem:first-component-simplified}
For matrix multiplication over rings, it is possible to partition $\TTall$ into $\TTA \cup \TTB$ such that
\begin{enumerate}
    \item $\TTA$ can be processed in $O(d^{1.858})$ rounds in the supported low-bandwidth model, and
    \item $\TTB$ contains at most $d^{1.9}n$ triangles.
\end{enumerate}
\end{lemma}
\begin{proof}
    We will solve the case of a small $d$ by brute force; hence let us assume that $d$ is sufficiently large.
    We apply \cref{lem:clustering} to $\TT = \TTall$ with $\eps_1 = 0$, $\eps_2 = 0.1$, and $\delta = 0.05$.
    We will set $\TTA = \PP^1 \cup \dotsb \cup \PP^L$ and $\TTB = \TT'$.
    The size of $\TTB$ is bounded by $d^{1.9}n$.
    The number of layers is $L = O(d^{0.7})$.
    
    We can now apply \cref{lem:intro-clustered} for $L$ times to process $P^1, \dotsc, P^L$; each step takes $O(d^{\denseexp})$ rounds and the total running time will be therefore $O(d^{1.858})$ rounds.
\end{proof}

We will later in \cref{sec:sparse-case} see how one could then process $\TTB$ in $O(d^{1.95})$ rounds, for a total running time of $O(d^{1.95})$ rounds.
However, we can do slightly better, as we will see next.

\subsection{Final algorithm for rings}

We can improve the parameters of \cref{lem:first-component-simplified} a bit, and obtain the following result. Here we have chosen the parameters so that the time we take now to process $\TTA$ will be equal to the time we will eventually need to process $\TTB$.

\begin{lemma}\label{lem:first-component-final}
For matrix multiplication over rings, it is possible to partition $\TTall$ into $\TTA \cup \TTB$ such that
\begin{enumerate}
    \item $\TTA$ can be processed in $O(d^{\mainresult})$ rounds in the supported low-bandwidth model, and
    \item $\TTB$ contains at most $d^{\interfacedensity}n$ triangles.
\end{enumerate}
\end{lemma}

\begin{table}
    \caption{Proof of \cref{lem:first-component-final}: parameters for the five iterations of \cref{lem:clustering}.}
	\vspace{0.5em}
    \centering
    \begin{tabular}{ccccc}
    \toprule
    Iteration & $\eps_1$ & $\eps_2$ & $\delta$ & $T_c$ \\
    \midrule
    1 & 0        & 0.149775 & 0.00001 & $O(d^{1.906016})$ \\
    2 & 0.149775 & 0.179736 & 0.00001 & $O(d^{1.906044})$ \\
    3 & 0.179736 & 0.185724 & 0.00001 & $O(d^{1.906024})$ \\
    4 & 0.185724 & 0.186926 & 0.00001 & $O(d^{1.906044})$ \\
    5 & 0.186926 & 0.187166 & 0.00001 & $O(d^{1.906044})$ \\
    \bottomrule
    \end{tabular}
    \label{tab:parameters}
\end{table}

\begin{table}
    \caption{Proof of \cref{lem:first-component-final-sr}: parameters for the five iterations of \cref{lem:clustering}.}
	\vspace{0.5em}
    \centering
    \begin{tabular}{ccccc}
    \toprule
    Iteration & $\eps_1$ & $\eps_2$ & $\delta$ & $T_c$ \\
    \midrule
    1 & 0        & 0.118537 & 0.00001 & $O(d^{1.926026})$ \\
    2 & 0.118537 & 0.142249 & 0.00001 & $O(d^{1.926050})$ \\
    3 & 0.142249 & 0.146986 & 0.00001 & $O(d^{1.926020})$ \\
    4 & 0.146986 & 0.147937 & 0.00001 & $O(d^{1.926040})$ \\
    5 & 0.147937 & 0.148127 & 0.00001 & $O(d^{1.926040})$ \\
    \bottomrule
    \end{tabular}
    \label{tab:parameters-sr}
\end{table}

\begin{proof}
    We apply \cref{lem:clustering} iteratively with the parameters given in \cref{tab:parameters}. Let $\eps_1^i$,$\eps_2^i,$ and $\delta_i$ be the respective parameters in iteration $i$. After iteration $i$, we process all triangles in $\PP^1 \cup \dotsb \cup \PP^L$, in the same way as in the proof of \cref{lem:first-component-simplified}. Then we are left with at most $d^{2-\eps_2^i}n$ triangles. We can then set $\eps_1^{i+1} = \eps_2^i$, and repeat.

    In \cref{tab:parameters}, $T_c$ is the number of rounds required to process the triangles; the total running time is bounded by $O(d^{\mainresult})$, and the number of triangles that are still unprocessed after the final iteration is bounded by $d^{2-0.187166} n < d^{\interfacedensity} n$.
\end{proof}

\subsection{Final algorithm for semirings}

In the proof of \cref{lem:first-component-final} we made use of \cref{lem:intro-clustered}, which is applicable for matrix multiplication over rings. If we plug in \cref{lem:intro-clustered-sr} and use the parameters from \cref{tab:parameters-sr}, we get a similar result for semirings:

\begin{lemma}\label{lem:first-component-final-sr}
For matrix multiplication over semirings, it is possible to partition $\TTall$ into $\TTA \cup \TTB$ such that
\begin{enumerate}
    \item $\TTA$ can be processed in $O(d^{\mainresultsr})$ rounds in the supported low-bandwidth model, and
    \item $\TTB$ contains at most $d^{\interfacedensitysr}n$ triangles.
\end{enumerate}
\end{lemma}

\section{Handling the small component}
\label{sec:sparse-case}

It remains to show how to process the triangles in the remaining small component $\TTB$. 

\begin{lemma}\label{lem:sparse-case-main}
Let $d \ge 2$ and $\eps < 1$. Let $\TT \subseteq \TTall$ be a collection of triangles with
\[
    \bigl|\TT\bigr| = O(d^{2-\eps} n).
\]
Then all triangles in $\TT$ can be processed in $O(d^{2-\eps/2})$ rounds in the supported low-bandwidth model, both for matrix multiplication over rings and semirings.
\end{lemma}

We first show how to handle the \emph{uniformly sparse} case, where we have a non-trivial bound on the number of triangles touching each node. We then show how reduce the small component case of $O(d^{2-\eps}n)$ triangles in total to the uniformly sparse case with each node touching at most $O(d^{2-\eps/2})$ triangles.

\subsection{Handling the uniformly sparse case}
\label{sec:usparse}

To handle the uniformly sparse cases, we use a simple brute-force algorithm. Note that setting $t = d^2$ here gives the trivial $O(d^2)$-round upper bound.

\begin{lemma}
\label{lem:brute-force}
Assume each node $i \in V$ is included in at most $t$ triangles in $\TT$. Then all triangles in $\TT$ can be processed in $O(t)$ rounds, both for matrix multiplication over rings and semirings.
\end{lemma}

\begin{proof}
To process all triangles in $\TT$, we want, for each triangle $\{ i, j, k \} \in \TT$, that node $j$ sends the entry $B_{jk}$ to node $i$. Since node $i$ knows $A_{ij}$ and is responsible for output $X_{ik}$, this allows node $i$ to accumulate products $A_{ij}B_{jk}$ to $X_{ik}$.

As each node is included in at most $t$ triangles, this means that each node has at most $t$ messages to send, and at most $t$ messages to receive, and all nodes can compute the sources and destinations of all messages from $\TT$. To see that these messages can be delivered in $O(t)$ rounds, consider a graph with node set $V \times V$ and edges representing source and destination pairs of the messages. Since this graph has maximum degree $t$, it has an $O(t)$-edge coloring (which we can precompute in the supported model). All nodes then send their messages with color $k$ on round $k$ directly to the receiver, which takes $O(t)$ rounds.
\end{proof}

\subsection{From small to uniformly sparse}

We now proceed to prove \cref{lem:sparse-case-main}. For purely technical convenience, we assume that $\TT$ contains at most $d^{2-\eps} n$ triangles---in the more general case of \cref{lem:sparse-case-main}, one can for example split $\TT$ into constantly many sets of size at most $d^{2-\eps} n$ and run the following algorithm separately for each one.

\paragraph{Setup} We say that a node $i \in V$ is a \emph{bad node} if $i$ is included in at least $d^{2-\eps/2}$ triangles in $\TT$, and we say that a triangle is a \emph{bad triangle} if it includes a bad node.
Since every triangle touches at most 3 nodes, the number of bad nodes is at most $3d^{2-\eps}n/d^{2-\eps/2} =  3n/d^{\eps/2}$.

We now want to distribute the responsibility for handling the bad triangles among multiple nodes. The following lemma gives us the formal tool for splitting the set of bad triangles evenly.

\begin{lemma}\label{lemma:bad-triangle-colouring}	
Assume $d \ge 2$ and $\eps < 1/2$. There exists a coloring of the bad triangles with $d^{\eps/2}/3$ colors such that for any color $c$, each bad node $i \in V$ touches at most $6d^{2-\eps/2}$ bad triangles of color $c$.
\end{lemma}

\begin{proof}
We prove the existence of the desired coloring by probabilistic method. We color each triangle uniformly at random with $d^{\eps/2}/3$ colors. Let $A_{v,c}$ denote the event that there are at least $6d^{2-\eps/2}$ triangles of color $c$ touching a bad node $v$ in this coloring; these are the bad events we want to avoid.

Fixing $v$ and $c$, let $X_{v,c}$ be the random variable counting the number of triangles of color $c$ touching $v$. This is clearly binomially distributed with the number of samples corresponding to the number of triangles touching $v$ and $p = 3/d^{\eps/2}$. 

Let $t_v$ be the number of triangles touching $v$. We have that
\[ d^{2-\eps/2} \leq t_v \leq d^2\,, \]
and thus
\[ 3d^{2-\eps} \leq \E[X_{v,c}] \leq  3d^{2 - \eps/2}\,.\]
By Chernoff bound, it follows that
\[ \Pr[ A_{v,c}] \le \Pr[ X_{v,c} \ge 2 \E[X_{v,c}]] \le e^{-4\E[X_{v,c}] / 3} \le e^{-4d^{2-\eps}}\,.\]

We now want to apply Lov\'as Local Lemma to show that the probability that none of the events $A_{v,c}$ happen is positive. 
We first observe that event $A_{v,c}$ is trivially independent of any set of other events $A_{u,d}$ that does not involve any neighbor of $v$. Since the degree of graph $G$ is at most $d$, the degree of dependency (see e.g.\ \cite{jukna2011extremal}) for events $A_{v,c}$ is at most $d^{\eps/2}(d+1)/3$.

It remains to show that Lov\'as Local Lemma condition $ep(D+1) \le 1$ is satisfied, where $p$ is the upper bound for the probability of a single bad event happening, and $D$ is the degree of dependency. Since we assume that $d \ge 2$ and $\eps < 1/2$, we have
\begin{align*}
1/pe & = e^{4d^{2-\eps}-1} > e^{2d^{2-\eps}} \ge 1 + (e-1)2d^{2-\eps} \\
     & \ge 1 + 2d^{2-\eps} > 1 + 2d^{1+\eps/2} \ge 1 + (d+1)d^{\eps/2} \ge D + 1\,.
\end{align*}
The claim now follows from Lov\'as Local Lemma.
\end{proof}

\paragraph{Virtual instance.} We now construct a new \emph{virtual} node set
\[\virt{V} = \virt{I} \cup \virt{J} \cup \virt{K}\,,\]
along with matrices $\virt{A}$ and $\virt{B}$ indexed by $\virt{I}$, $\virt{J}$, and $\virt{K}$, and a set of triangles $\virt{\TT}$ we need to process from this virtual instance. The goal is that processing all triangles in $\virt{\TT}$ allows us to recover the solution to the original instance.

Let $\chi$ be a coloring of bad triangles as per \cref{lemma:bad-triangle-colouring}, and let $C$ be the set of colors used by $\chi$. We construct the virtual node set $\virt{I}$ (resp., $\virt{J}$ and $\virt{K}$) by adding a node $i_0$ for each non-bad node $i \in I$ (resp., $j \in J$ and $k \in K$), and nodes $v_c$ for bad node $v \in V$ and color $c \in C$. Note that the virtual node set $\virt{V}$ has size at most $2 |V|$. Finally, for technical convenience, we define a \emph{color set} $c(v)$ for node $v \in V$ as
\[
c(v) = 
\begin{cases}
\{ 0 \}\,,&  \text{if $v$ is non-bad, and}\\
C\,,      & \text{if $v$ is bad.}
\end{cases}
\]
The matrix $\virt{A}$ is now defined by setting $A_{i_c j_d} = A_{ij}$ for $i \in I$, $j \in J$ and colors $c \in c(i)$ and $d \in c(j)$. Matrix $\virt{B}$ is defined analogously.

The set of triangles $\virt{\TT}$ is constructed as follows. For each non-bad triangle  $\{ i, j,k \} \in \TT$, we add $\{ i_0, j_0, k_0 \}$ to $\virt{\TT}$. For each bad triangle $T \in \TT$ with color $\chi(T)$, we add to $\virt{\TT}$ a new triangle obtained from $T$ by replacing each bad node $v \in T$ with $i_{\chi(\TT)}$ and each non-bad node by $v_0$. By construction, each node in $\virt{V}$ is included in at most $6d^{2-\eps/2}$ triangles in $\virt{\TT}$.  Moreover, if matrix $\virt{X}$ represents the result of processing all triangles in $\virt{\TT}$, we can recover the results $X$ of processing triangles in $\TT$ as
\begin{equation}\label{eq:simulation-recovery}
X_{ij} = \sum_{c \in c(i)} \sum_{d \in c(j)} \virt{X}_{i_c j_d}\,.
\end{equation}

\paragraph{Simulation.} We now show that we can construct the virtual instance as defined above from $A$, $B$ and $\TT$, simulate the execution of the algorithm of \cref{lem:brute-force} in the virtual instance with constant-factor overhead in the round complexity, and recover the output of the original instance from the result.

As preprocessing based on the knowledge of $\TT$, all nodes perform the following steps locally in an arbitrary and consistent way:
\begin{enumerate}
	\item Compute the set of bad nodes and bad triangles.
	\item Compute a coloring $\chi$ of bad triangles as per \cref{lemma:bad-triangle-colouring}.
	\item For each bad node $v \in V$, assign a set of \emph{helper nodes}
	\[ U_v = \{ v_c \colon c \in C \} \subseteq V \]
	of size $d^{\eps/2}/2$ so that helper node sets are disjoint for distinct bad nodes. Note that this is always possible, as the number of required nodes is at most $|V|$.
\end{enumerate}

We now simulate the execution of triangle processing on the virtual instance as follows. Each non-bad node $v$ simulates the corresponding node $v_0$ in the virtual instance, as well as a duplicate of bad node $u_c$ if they are assigned as helper node $v_c$. Each bad node simulates their assigned bad node duplicate. To handle the duplication of the inputs and collection of outputs, we use the following simple routing lemma.

\begin{lemma}\label{lemma:broadcast-and-convercast}
Let $v \in V$ be a node and let $U \subseteq V$ be a set of $k$ nodes. The following communication tasks can be performed in $O(d + \log k)$ rounds in low-bandwidth model using only communication between nodes in $\{ v \} \cup U$:
\begin{enumerate}
	\item[(a)] Node $v$ holds $d$ messages of $O(\log n)$ bits, and each node in $U$ needs to receive each message held by $v$.
	\item[(b)] Each node $u \in U$ holds $d$ values $s_{u,1}, \dotsc s_{u,d}$, and node $v$ needs to learn $\sum_{u \in U} v_{u,i}$ for $i = 1, 2, \dotsc, d$.
\end{enumerate}
\end{lemma}

\begin{proof}
For both parts, fix an arbitrary binary tree $T$ on $\{ v \} \cup U$ rooted at $v$. For part (a), a single message can be broadcast to all nodes along $T$ in $O(\log k)$ rounds, by simply having each node spend $2$ rounds sending it to both its children once the node receives the message. For $d$ messages, we observe that the root $v$ can start the sending the next message immediately after it has sent the previous one---in the standard pipelining fashion---and the communication for these messages does not overlap. Thus, the last message can be sent by $v$ in $O(d)$ rounds, and is received by all nodes in $O(d + \log k)$ rounds.

For part (b), we use the same idea in reverse. For a single index $i$, all leaf nodes $u$ send $s_{u,i}$ to their parent, alternating between left and right children on even and odd rounds. Subsequently, nodes compute the sum of values they received from their children, and send it their parents. For multiple values, the pipelining argument is identical to part~(a).
\end{proof}

Now the simulation proceeds as follows:
\begin{enumerate}
	\item Each bad node $v \in V$ sends their input to all nodes in $U_v$ in parallel. This takes $O(d^{\eps/2} + \log d)$ rounds by \cref{lemma:broadcast-and-convercast}(a).
	\item Each node locally computes the rows of $\virt{A}$ and $\virt{B}$ for the nodes of $\virt{V}$ they are responsible for simulating.
	\item Nodes collectively simulate the algorithm of \cref{lemma:bad-triangle-colouring} on the virtual instance formed by $\virt{A}$, $\virt{B}$ and $\virt{\TT}$ to process all triangles in the virtual triangle set $\virt{\TT}$. Since all nodes in $\virt{V}$ touch $O(d^{2-\eps/2})$ triangles, and overhead from simulation is $O(1)$, this takes $O(d^{2-\eps/2})$ rounds.
	\item Each bad node $i \in I$ recovers row $i$ of output $X$ according to \cref{eq:simulation-recovery} by using \cref{lemma:broadcast-and-convercast}(b). Each non-bad node $i$ computes the row $i$ of $X$ based on their local knowledge according to \cref{eq:simulation-recovery}. This takes $O(d^{\eps/2} + \log d)$.
\end{enumerate}

\section{Putting things together}

Now we are ready to prove our main theorems:

\begin{proof}[Proof of \cref{thm:main-intro}]
By \cref{lem:first-component-final}, we can partition $\TTall$ into $\TTA$ and $\TTB$ such that $\TTA$ can be processed in $O(d^{\mainresult})$ rounds, and $\TTB$ has at most $d^{\interfacedensity} n$ triangles. Then we can apply \cref{lem:sparse-case-main} to $\TTB$ with $\eps = \epschoice$ and hence process also $\TTB$ in $O(d^{\mainresult})$ rounds.
\end{proof}

\begin{proof}[Proof of \cref{thm:main-intro-sr}]
By \cref{lem:first-component-final-sr}, we can partition $\TTall$ into $\TTA$ and $\TTB$ such that $\TTA$ can be processed in $O(d^{\mainresultsr})$ rounds, and $\TTB$ has at most $d^{\interfacedensitysr} n$ triangles. Then we can apply \cref{lem:sparse-case-main} to $\TTB$ with $\eps = \epschoicesr$ and hence process also $\TTB$ in $O(d^{\mainresultsr})$ rounds.
\end{proof}

\subsection*{Acknowledgements}

We are grateful to the anonymous reviewers for their helpful feedback on the previous versions of this work.
This work was supported in part by the Academy of Finland, Grant 321901.

\bibliographystyle{plainnat}
\bibliography{triangles}

\begin{thebibliography}{20}
\providecommand{\natexlab}[1]{#1}
\providecommand{\url}[1]{\texttt{#1}}
\expandafter\ifx\csname urlstyle\endcsname\relax
  \providecommand{\doi}[1]{doi: #1}\else
  \providecommand{\doi}{doi: \begingroup \urlstyle{rm}\Url}\fi

\bibitem[Alman and Williams(2021)]{alman2021refined}
Josh Alman and Virginia~Vassilevska Williams.
\newblock A refined laser method and faster matrix multiplication.
\newblock In \emph{Proc.\ ACM-SIAM Symposium on Discrete Algorithms (SODA
  2021)}, pages 522--539, 2021.
\newblock \doi{10.1137/1.9781611976465.32}.

\bibitem[Augustine et~al.(2019)Augustine, Ghaffari, Gmyr, Hinnenthal,
  Scheideler, Kuhn, and Li]{node-capacitated2019}
John Augustine, Mohsen Ghaffari, Robert Gmyr, Kristian Hinnenthal, Christian
  Scheideler, Fabian Kuhn, and Jason Li.
\newblock Distributed computation in node-capacitated networks.
\newblock In \emph{Proc.\ 31st ACM Symposium on Parallelism in Algorithms and
  Architectures (SPAA 2019)}, page 69–79. ACM, 2019.
\newblock \doi{10.1145/3323165.3323195}.

\bibitem[B{\"u}rgisser et~al.(1997)B{\"u}rgisser, Clausen, and
  Shokrollahi]{burgisser-algebraic}
Peter B{\"u}rgisser, Michael Clausen, and Mohammad~A Shokrollahi.
\newblock \emph{Algebraic complexity theory}.
\newblock 1997.

\bibitem[Censor-Hillel et~al.(2018)Censor-Hillel, Leitersdorf, and
  Turner]{DBLP:conf/opodis/Censor-HillelLT18}
Keren Censor-Hillel, Dean Leitersdorf, and Elia Turner.
\newblock Sparse matrix multiplication and triangle listing in the congested
  clique model.
\newblock In \emph{Proc.\ {OPODIS} 2018}, 2018.
\newblock \doi{10.4230/LIPIcs.OPODIS.2018.4}.

\bibitem[Censor-Hillel et~al.(2019)Censor-Hillel, Kaski, Korhonen, Lenzen, Paz,
  and Suomela]{alg-method-congest-cliq2019}
Keren Censor-Hillel, Petteri Kaski, Janne~H. Korhonen, Christoph Lenzen, Ami
  Paz, and Jukka Suomela.
\newblock Algebraic methods in the congested clique.
\newblock \emph{Distributed Comput.}, 32\penalty0 (6):\penalty0 461--478, 2019.
\newblock \doi{10.1007/s00446-016-0270-2}.

\bibitem[Censor-Hillel et~al.(2020)Censor-Hillel, Gall, and
  Leitersdorf]{10.1145/3382734.3405742}
Keren Censor-Hillel, Fran\c{c}ois~Le Gall, and Dean Leitersdorf.
\newblock On distributed listing of cliques.
\newblock In \emph{Proc.\ 39th Symposium on Principles of Distributed Computing
  (PODC 2020)}, page 474–482, 2020.
\newblock \doi{10.1145/3382734.3405742}.

\bibitem[Censor-Hillel et~al.(2021)Censor-Hillel, Dory, Korhonen, and
  Leitersdorf]{censor2021fast}
Keren Censor-Hillel, Michal Dory, Janne~H Korhonen, and Dean Leitersdorf.
\newblock Fast approximate shortest paths in the congested clique.
\newblock \emph{Distributed Computing}, 34\penalty0 (6):\penalty0 463--487,
  2021.
\newblock \doi{10.1007/s00446-020-00380-5}.

\bibitem[Chang and Saranurak(2019)]{chang2019improved}
Yi-Jun Chang and Thatchaphol Saranurak.
\newblock Improved distributed expander decomposition and nearly optimal
  triangle enumeration.
\newblock In \emph{Proc.\ 38th ACM Symposium on Principles of Distributed
  Computing (PODC 2019)}, pages 66--73, 2019.

\bibitem[Chang et~al.(2021)Chang, Pettie, Saranurak, and
  Zhang]{10.1145/3446330}
Yi-Jun Chang, Seth Pettie, Thatchaphol Saranurak, and Hengjie Zhang.
\newblock Near-optimal distributed triangle enumeration via expander
  decompositions.
\newblock \emph{Journal of the ACM}, 68\penalty0 (3), 2021.
\newblock \doi{10.1145/3446330}.

\bibitem[Dolev et~al.(2012)Dolev, Lenzen, and Peled]{dolev2012tri}
Danny Dolev, Christoph Lenzen, and Shir Peled.
\newblock “tri, tri again”: Finding triangles and small subgraphs in a
  distributed setting.
\newblock In \emph{Proc.\ {DISC} 2012}, pages 195--209, 2012.

\bibitem[Foerster et~al.(2019)Foerster, Korhonen, Rybicki, and
  Schmid]{foerster2019preprocessing}
Klaus-Tycho Foerster, Janne~H. Korhonen, Joel Rybicki, and Stefan Schmid.
\newblock Does preprocessing help under congestion?
\newblock In \emph{Proc.\ 38nd {ACM} Symposium on Principles of Distributed
  Computing, ({PODC} 2019)}, pages 259--261, 2019.
\newblock \doi{10.1145/3293611.3331581}.

\bibitem[Izumi and Le~Gall(2017)]{izumi2017triangle}
Taisuke Izumi and Fran{\c{c}}ois Le~Gall.
\newblock Triangle finding and listing in {CONGEST} networks.
\newblock In \emph{Proc.\ 36th ACM Symposium on Principles of Distributed
  Computing (PODC 2017)}, pages 381--389, 2017.

\bibitem[Jukna(2011)]{jukna2011extremal}
Stasys Jukna.
\newblock \emph{Extremal combinatorics: with applications in computer science}.
\newblock Springer, 2011.

\bibitem[Korhonen and Rybicki(2017)]{korhonen2017deterministic}
Janne~H. Korhonen and Joel Rybicki.
\newblock Deterministic subgraph detection in broadcast {CONGEST}.
\newblock In \emph{Proc.\ 21st International Conference on Principles of
  Distributed Systems ({OPODIS} 2017)}, 2017.
\newblock \doi{10.4230/LIPIcs.OPODIS.2017.4}.

\bibitem[Le~Gall(2016)]{le2016further}
Fran{\c{c}}ois Le~Gall.
\newblock Further algebraic algorithms in the congested clique model and
  applications to graph-theoretic problems.
\newblock In \emph{Proc.\ {DISC} 2016}, pages 57--70, 2016.

\bibitem[Lotker et~al.(2003)Lotker, Pavlov, Patt-Shamir, and
  Peleg]{congest-clique2003}
Zvi Lotker, Elan Pavlov, Boaz Patt-Shamir, and David Peleg.
\newblock Mst construction in o(log log n) communication rounds.
\newblock In \emph{Proc.\ 15th Annual ACM Symposium on Parallel Algorithms and
  Architectures (SPAA 2003)}, page 94–100. ACM, 2003.
\newblock \doi{10.1145/777412.777428}.

\bibitem[Pandurangan et~al.(2021)Pandurangan, Robinson, and
  Scquizzato]{10.1145/3460900}
Gopal Pandurangan, Peter Robinson, and Michele Scquizzato.
\newblock On the distributed complexity of large-scale graph computations.
\newblock \emph{ACM Transactions on Parallel Computing}, 8\penalty0 (2), 2021.
\newblock \doi{10.1145/3460900}.

\bibitem[Strassen(1969)]{strassen1969gaussian}
Strassen.
\newblock Gaussian elimination is not optimal.
\newblock \emph{Numerische Mathematik}, 13\penalty0 (4), 1969.

\bibitem[Valiant(1990)]{BSP1990}
Leslie~G. Valiant.
\newblock A bridging model for parallel computation.
\newblock \emph{Commun. ACM}, 33\penalty0 (8):\penalty0 103–111, 1990.
\newblock \doi{10.1145/79173.79181}.

\bibitem[Yuster and Zwick(2005)]{10.1145/1077464.1077466}
Raphael Yuster and Uri Zwick.
\newblock Fast sparse matrix multiplication.
\newblock \emph{ACM Trans. Algorithms}, 1\penalty0 (1), 2005.
\newblock \doi{10.1145/1077464.1077466}.

\end{thebibliography}

\end{document}